\documentclass[12pt]{article}

\usepackage{amsmath}
\usepackage{graphicx}
\usepackage{color}
\usepackage{bbm}
\usepackage{epsfig}
\usepackage{latexsym}
\usepackage{color}

\newtheorem{theorem}{Theorem}
\newtheorem{lemma}{Lemma}
\newtheorem{corollary}{Corollary}

\newtheorem{definition}{Definition}

\begin{document}
\title{Superconcentration on a Pair of Butterflies}
\author{William F. Bradley}
\maketitle
\bibliographystyle{plain}

\begin{abstract}
 Suppose we concatenate two directed graphs, each isomorphic
  to a $d$ dimensional butterfly (but not necessarily identical to
  each other). Select any set of $2^{k}$ input and $2^{k}$ output
  nodes on the resulting graph.  Then there exist node disjoint paths
  from the input nodes to the output nodes.  If we take two standard
  butterflies and permute the order of the layers, then the result
  holds on sets of \emph{any} size, not just powers of two.
\end{abstract}

This paper will examine some problems in node-disjoint circuit switching.
The motivating problem can be described as follows.  
Suppose we have a directed graph with $N$ input
and $N$ output nodes, both labelled from 1 to $N$.  For each input node $v$,
we choose an output node $\pi(v)$ to be its destination, for some permutation
$\pi$.  The problem is to find a collection of $N$ node-disjoint paths 
which each run from $v$ to $\pi(v)$ for all $v$.  A directed graph that can
route all permutations $\pi$ is called \emph{rearrangeable}.
(For some real-world applications of
node-disjoint routing, see, for example,~\cite{optical_networks}.)

A classic example of rearrangeability is the Bene\v{s} 
network (see~\cite{Leighton}).  
This network (i.e.\ directed graph) consists of a ``forward'' butterfly 
adjoined to a 
``reversed'' butterfly.  A natural question to ask is: if we attach 
two ``forward'' butterflies, is this network (the 
double butterfly) still rearrangeable?
This problem has been open for several decades.  At least one proof is 
currently under review~\cite{Cam}.  This suggests a more general
hypothesis.  Suppose that we have two graphs, each isomorphic to a
butterfly, but not necessarily identical to each other.  If we attach
the output nodes of the first to the input nodes of the second, is 
the resulting graph rearrangeable?

At the current time, proving this kind of result seems far too much to
hope for.  So, rather than show that these types of networks are rearrangeable,
I will prove various concentration and superconcentration results.
\begin{definition}
Consider a directed graph $G$.  Fix $n$ input nodes and $n$ output nodes.
Suppose that between any $k$ input and $k$ output nodes there exist 
$k$ node-disjoint paths.
(By ``node-disjoint'', I mean that a path intersects
neither itself nor any other path.)  Then we say that $G$ is a
$k$-concentrator.  If $G$ is a $k$-concentrator
for all $k\leq n$, then we call $G$ a \emph{superconcentrator}.
\end{definition}
(Observe that every selected input and output node occurs on exactly one 
path.  Note also that node-disjointness implies edge-disjointness of the paths.
See \cite{Pippenger} or \cite{Hwang} for more on superconcentrators.)

Clearly, rearrangeability implies subset routing-- just choose a permutation
that respects $v\in A$ iff $\pi(v)\in B$.  However, the converse is not true
for arbitrary networks (see Figure~\ref{subset counterexample}).
\begin{figure}[ht]
\centerline{\psfig{file=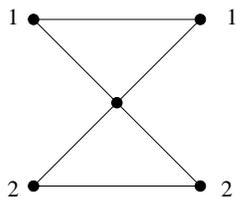,height=1in} \hspace{.2in}  }
\caption{A non-rearrangeable superconcentrator (consider 
$1\rightarrow 2$, $2\rightarrow 1$) }
\label{subset counterexample}
\end{figure}
It's straightforward to show that a single butterfly does not route all 
subsets, so we need to use at least two butterflies to get interesting 
concentration results.

In this paper, I show that any concatenated pair of $d$-dimensional 
butterflies 
(not necessarily identical to each other) are $2^{k}$-concentrators, 
for any $k\leq d$.  I can strengthen this statement in 
a special case:  if the butterflies
are standard butterflies with their layers shuffled (e.g. a Bene\v{s} 
network, or the double buttefly in~\cite{Cam}), the network
is a superconcentrator.

	This paper is structured as follows.  Section 1 establishes some 
definitions and fixes notation for the paper.  Section 2 examines the 
structure of a graph related to a pair of butterflies that highlights
some of its connectivity properties.  Section 3 solves the problem
in the case where $|A|=2^{m}$ for some $m$, and proves a 
rearrangeability-type result when $|A|\leq ^{\lfloor d/2 \rfloor}$
on certain networks.
Section 4 presents the main result, except for one lemma 
that I postpone for section 5.  Section 6 is dedicated to closing remarks.

\section{Definitions and Notation}

Let us begin by defining and fixing notation for a butterfly.
\begin{definition}
A \emph{(standard) $d$-dimensional butterfly} 
is a directed, layered graph defined as follows: 
nodes fall into one of $d+1$ disjoint layers, numbered 0 through $d$. Each
layer consists of $N=2^d$ nodes, which we label with the $N$ binary
strings of length $d$.  (So, a node is specified by a binary string and a
layer number.)
Consider any length $d$ binary string, say $b=b_{1}b_{2}\cdots b_{d}$.
For each $i$ such that $0\leq i <d$,
there is a directed edge from node $b$ of layer $i$ to node 
$b_{1}b_{2}\cdots b_{i-1}0 b_{i+1}\cdots b_{d}$ of layer $i+1$, and 
another directed edge from node $b$ of layer $i$ to node 
$b_{1}b_{2}\cdots b_{i-1}1 b_{i+1}\cdots b_{d}$ of layer $i+1$.
The nodes on layer 0 are called the input nodes, and the nodes on layer 
$d$ are the output nodes.
\end{definition}
A standard $d$-dimensional butterfly can be viewed as a network with 
$2^{d}$ nodes where we switch the first bit in the first layer
of edges, the second bit in the second layer of edges, and so forth.
If we choose to switch the bits in a different order, we get a
\emph{layer-permuted butterfly}.
\begin{definition}
A \emph{$d$-dimensional layer-permuted butterfly} 
is a directed, layered graph defined as follows: 
nodes fall into one of $d+1$ disjoint layers, numbered 0 through $d$. Each
layer consists of $N=2^d$ nodes, which we label with the $N$ binary
strings of length $d$.  Take some (fixed) permutation $\pi$ on $d$
objects.  Consider any such binary string, say $b=b_{1}b_{2}\cdots b_{d}$.
For each $i$ such that $0\leq i <d$,
there is a directed edge from node $b$ of layer $i$ to node 
$b_{1}b_{2}\cdots b_{\pi(i)-1}0 b_{\pi(i)+1}\cdots b_{d}$ of layer $i+1$, and 
another directed edge from node $b$ of layer $i$ to node 
$b_{1}b_{2}\cdots b_{\pi(i)-1}1 b_{\pi(i)+1}\cdots b_{d}$ of layer $i+1$.
\end{definition}
(Note that these butterflies are all graph-isomorphic to each other.)
Finally, the networks we'll be looking at consist of pairs of these
butterflies.
\begin{definition}
Suppose that we have two graphs, $G_{1}$ and $G_{2}$.  Suppose that 
$G_{1}$ has $n$ output nodes and $G_{2}$ has $n$ input nodes, each
numbered from 1 to $n$.  Then
we say that $G$ is the \emph{concatenation} of $G_{1}$ and $G_{2}$ if 
we form $G$ by associating the output node $i$ of $G_{1}$ with 
the input node $i$ of $G_{2}$.  
 
If $G_{1}$ and $G_{2}$ are each isomorphic to a standard butterfly
(but not necessarily identical to each other), we call $G$ 
a \emph{pair of butterflies}.  Similarly, if $G_{1}$ and $G_{2}$ are
layer-permuted butterflies, we have a 
\emph{pair of layer-permuted butterflies}.  Finally, if $G_{1}$ and 
$G_{2}$ are both standard butterflies, we have a \emph{double butterfly}.
\end{definition}
Note that these graphs have $2d+1$ layers of
nodes (0 through $2d$).  Since I imagine the paths from inputs to outputs
to be running from 
left to right, I will refer to the butterfly on layers 0 through $d$ as
the \emph{left butterfly}, and the one on layers $d$ through $2d$ as the
\emph{right butterfly}.  Note also an alternate way of specifying
a pair of butterflies: consider a network consisting of
two standard butterflies, but permute the labels of the output nodes
of the left butterfly.  
Observe that these two definitions give rise to the same class of 
graphs (up to isomorphism).

Over the course of this paper, I construct directed node-disjoint 
paths from input nodes to output nodes.  So, for example, a path from an 
input
node to an output node on a pair of butterflies is exactly $2d+1$ nodes
long-- the path can't double back on the layers.  

Suppose I have a set of input nodes $A$ and a set of output nodes $B$ of the
same size (i.e.\ $|A|=|B|$)
Then if I specify a collection of node-disjoint paths from $A$ to $B$, observe
that I can extend these paths into a consistent setting of all the 
switches in the network.  These switches will induce $N$  node disjoint
paths from \emph{every} input to every output node, and retain the
feature that a path begins in $A$ iff it ends in $B$.  So, on a
switching network, node-disjoint routing of a subset implies there
exists a node-disjoint routing of a permutation $\pi$ such that 
$v\in A$ iff $\pi(v)\in B$.  Since this is an ``if and only if'' statement,
we get the following lemma:
\begin{lemma} \label{set complement}
If we can find node-disjoint paths from $A$ to $B$ on a switching
network, then we can find node-disjoint paths from the complements
$A^{c}$ to $B^{c}$.
\end{lemma}
Throughout this paper, I will use $A$ to represent a collection of
input nodes, $B$ a collection of output nodes, and assume that
$|A|=|B|$.

\section{The Sub-Butterfly Connectivity Graph} \label{connectivity}

Suppose we specify a path of length $m$ on a standard butterfly 
(for $m\leq d$) from an input node.
By choosing which edge to take, the path changes $m$ bits of its
location any way we want.  Suppose we select the first bits to be
$b=b_{1}b_{2}b_{3}\cdots b_{m}$.  Then from layers $m+1$ to $d$, 
the first $m$ bits will remain equal to $b$.  Let's specify 
the resulting sub-graph of the butterfly in the following definition:

\begin{definition}
Consider a $d$ dimensional standard butterfly.  Take an $m$ bit
binary string $b=b_{1}b_{2}b_{3}\cdots b_{m}$  ($m\leq d$).
Consider the sub-graph formed by the nodes on layers $m$ through $d$
(inclusive) whose first $m$ bits are $b$.  Observe that this graph
is (isomorphic to) a $(d-m)$-dimensional butterfly.  Let us call it 
the \emph{sub-butterfly $b*$}.

If we specify a suffix instead and consider layers 0 through $d-m$,
we get the \emph{sub-butterfly $*b$}.

If we have a graph isomorphic to a standard butterfly, the isomorphism
will induce (isomorphic) images of the sub-butterfly, so we can
meaningfully refer to sub-butterflies on any butterfly-isomorphic graph.
\end{definition}
I will be considering sub-butterflies in a pair of butterflies.  In
this context, $b*$ is the sub-butterfly residing on layers $m$
through $d$ (and stopping there), i.e.\ only in the left butterfly.
I'll also be interested in sub-butterflies on the right side.  
These inhabit layers $d$ through $2d-m$.

It will be useful to investigate the structure of the connections between
the $q$-dimensional sub-butterflies on the right and left sides of
a $d$-dimensional pair of butterflies, that is, the sub-butterflies of the
 form $x*$ or $*x$ where $x$ is a binary string of length $m$ (such that
$m+q=d$).  Note that these sub-butterflies inhabit layers $m=d-q$ through
 $d$, and $d$ through $d+q$.  Let us represent each 
sub-butterfly by a vertex in a bipartite graph; the vertex is on the left
side of the bipartite graph iff the sub-butterfly is on the left
side of the pair of butterflies.  I will label each vertex by its associated
sub-butterfly, abusing the label notation somewhat.  
Place an edge between two vertices $x*$ and $*y$ iff
the two sub-butterflies are connected, that is, iff $x*$ and $*y$
(as sub-butterflies) share at least
one common node on layer $d$ of the pair of butterflies.
Equivalently, there is an edge between the nodes in the bipartite graph
iff there exists a path from every layer $d-q$ input node 
of $x*$ to every layer $d+q$ output node of $*y$. 
I will refer to this graph as the \emph{$q$-dimensional 
sub-butterfly connectivity graph},
or just the \emph{connectivity graph}.
Observe that there are $2^{m}$ vertices
on either side of this graph.  How are the vertices connected?

I will consider progressively more specialized cases in order to derive 
various results in later sections.  Suppose, first, that we build a bipartite
connectivity graph, but if there are $x$ common nodes on layer $d$ 
between a sub-butterfly on the left and one on the right, 
we insert $x$ edges (instead of only 1 edge).  Let us call this the 
\emph{enriched connectivity graph}.  

\begin{lemma} \label{enriched graph is regular}
For any pair of butterflies, its enriched connectivity graph is regular.
\end{lemma}

\begin{proof} Since each sub-butterfly has $2^{q}$ output nodes, then all 
nodes in the enriched connectivity graph have degree $2^{q}$.
\end{proof}

Now we move our attention to the special case of layer-permuted butterflies.
First, let us analyze the structure of one connected compnent of
the $q$-dimensional connectivity graph.

\begin{lemma} \label{connected components}
Each connected component in the connectivity graph of a layer-permuted
butterfly is a completely connected bipartite graph.
\end{lemma}

\begin{proof}
Suppose that the layer-permuted butterfly on the left has permutation
$\pi$, and the butterfly on the right has permutation $\sigma$.
Consider a sub-butterfly $b*$ in the left butterfly.  This corresponds to 
a sub-graph on layers $m$ through $d$ where the value of bit $\pi(i)$ 
is $b_{i}$.  Notice that a sub-butterfly $b*$ in the left butterfly 
connects to 
a  sub-butterfly $*c$ in the right butterfly if and only if
\begin{equation} \label{sub-butterfly connection}
\forall i < q, \forall j > m, \mbox{ if }\pi(i)=\sigma(j) 
\mbox{ then } b_{i} = c_{j} 
\end{equation}
Thus, each connected 
component is a complete bipartite graph (with the same
number of nodes on each side.)
\end{proof}

Next, suppose that we have a pair of layer-permuted butterflies.  
How does the graph change as we specify one more
layer?  That is, if we compare the connectivity graphs between $q$ 
and $q-1$ dimensional sub-butterflies, what happens?

Therefore, determining the structure of the connectivity graph on pairs
of layer-permuted butterflies reduces to determining the connected components.
Consider one connected component in the connectivity graph looking at
$q$-dimensional sub-butterflies.  When we advance to the 
$(q-1)$ dimensional sub-butterflies, each node becomes two nodes (because each
$q$ dimensional sub-butterfly splits into two $q-1$ dimensional 
sub-butterflies).  There are  essentially three cases that can occur.
\begin{itemize} \label{three cases}
\item \textbf{(No reused dimensions)}
Suppose that 
$\sigma(q-1) \neq \pi(j)$ for any $1 \leq j \leq m+1$ and
$\pi(m+1) \neq \sigma(j)$ for any $q-1 \leq j \leq d$.
Then if the $q$-dimensional sub-butterfly $b*$ is adjacent
to $*c$, it follows that $bb_{m_1}*$ is adjacent to 
$*c_{m+1}c$ for $b_{m-1}, c_{m-1} = 0,1$.  

In the connectivity graph, that means that the connected component doubles
the number of nodes, but remains completely connected.
\item \textbf{(One reused dimension)}
Suppose that there exists (exactly) one $i$ such that either 
\begin{itemize}
\item
$\sigma(q-1) = i = \pi(m+1)$, or
\item
$\sigma(q-1) = i = \pi(j)$ for some $1 \leq j \leq m+1$ and
$\pi(m+1) \neq \sigma(k)$ for any $q-1 \leq k \leq d$, or
\item
$\sigma(q-1) \neq \pi(j)$ for any $1 \leq j \leq m+1$ and
$\pi(m+1) = i = \sigma(k)$ for some $q-1 \leq k \leq d$
\end{itemize}
Then the connected component splits into two connected components,
based on the value of the $i$th bit.
\item  \textbf{(Two reused dimensions)}
Suppose that $\sigma(d-m-1) = i = \pi(j)$ for $1 \leq j \leq m+1$ and
$\pi(m+1) = l = \sigma(k)$ for $d-m-1 \leq k \leq d$, and $i \neq l$.
Then the connected component splits into four connected components, based
on the four possible values that the $i$ and $l$ bits can take.
\end{itemize}

\section{Subsets of size $2^m$}
We want to select a collection of node-disjoint paths from input set $A$
to output set $B$ on a pair of butterflies.
Although I've expressed this problem in terms of paths, it's often easier
to express the proof in terms of packets travelling through the network.
In particular, if packets travel forward (node disjointly, and
without stopping) from every input node in $A$,
and backwards from every output node in $B$, and we can match up the 
packets on level $d$, then the paths traced by the packets give us the
collection of paths we're looking for.  I will switch between the path
and packet descriptions of the problem whenever it seems helpful.
\begin{lemma} \label{simple splitting}
Suppose we have a set $A$ of input nodes on a butterfly.
By passing from layer 0 to layer 1 of a butterfly, there exist
paths that send
$\lceil |A|/2 \rceil$ of the packets to sub-butterfly $0*$, and 
$\lfloor |A|/2 \rfloor$ of the packets to sub-butterfly $1*$.
Similarly, we could send 
$\lceil |A|/2 \rceil$ of the packets to sub-butterfly $1*$, and 
$\lfloor |A|/2 \rfloor$ of the packets to sub-butterfly $0*$.
Mutatis mutandi, this applies to packets in output nodes
travelling backwards, by passing from layer $2d$ to $2d-1$.
\end{lemma}
\begin{proof}
The $N$ nodes on the first layer of the butterfly can be grouped into 
$N/2$ switches, where the nodes labelled $T_{0}=0t_{2}t_{3}\cdots t_{d}$
and $T_{1}=1t_{2}t_{3}\cdots t_{d}$ form one switch.  Observe that 
each switch can be set straight or crossed, that is, 
we have to send $T_{i}$ to $T_{i}$ on the next layer 
(for $i=$ both 0 and 1), or $T_{i}$ to $T_{1-i}$.  Setting switches in one of 
these two states guarantees that paths are node-disjoint, so I 
will always set them accordingly. 

	For all the switches such that $T_{0},T_{1}\in A$, half
of these packets get sent to $0*$, and half to $1*$.  If 
$T_{0},T_{1}\not\in A$, half of these (zero) packets get sent
to each sub-butterfly, too.  Consider all of the remaining
packets.  Each of these is the sole packet in the switch.  So,
by setting $\lceil |A|/2 \rceil$ of the switches to send the
packets to sub-butterfly $0*$, and $\lfloor |A|/2 \rfloor$ of them 
to $1*$, we prove the first part of the lemma.  The rest
follows by symmetry. 
\end{proof}

This lemma allows a surprisingly simple proof of $2^{m}$ concentration.
\begin{theorem} \label{powers of two}
Suppose $|A|=2^{m}=|B|$.  Then there exist node-disjoint paths
from any input set $A$ to any output set $B$ on a pair of butterflies.
(In other words, a pair of butterflies is a $2^{m}$-concentrator.)
\end{theorem}
\begin{proof}
Consider the left butterfly.
We can apply Lemma~\ref{simple splitting} recursively for $m$ steps.  On
step 1, we split $A$ so that $2^{m-1}$ packets go to $0*$ and 
$2^{m-1}$ go to $1*$.  Since $0*$ and $1*$ are themselves $d-1$ 
dimensional butterflies, we can apply the lemma again, on each
of them, giving us 4 sub-butterflies, each with $2^{m-2}$ 
paths.  After $m$ steps, we end up with $2^{m}$ sub-butterflies
(which is all of the $d-m$ dimensional sub-butterflies),
each of which has exactly 1 packet.  Now, on each of these butterflies,
we can send the packet along any path we want for the remainder
of the left butterfly (i.e.\ until we hit layer $d$); since it's the only 
packet on its sub-butterfly, there's no possibility of any other 
packet's path crossing its own.

We can perform the same construction on the output packets in $B$,
moving backwards toward the input layer.  When we reach layer $2d-m$,
there will be 1 packet per sub-butterfly.

At this point, observe that the sub-butterfly connectivity graph 
determines the connections between these butterflies.
By Lemma~\ref{enriched graph is regular},
this graph is a regular bipartite graph.
By Hall's theorem, there exists a 
perfect matching.  
This matching in the connectivity graph implies a matching in the
set of sub-butterflies, which implies a matching between the (unique)
packets in each sub-butterfly.  By construction of the connectivity
graph, there exists a path (not necessarily unique) between matched
packets.  As observed above, these paths are node-disjoint,
so we're done. 
\end{proof}

\subsection{Some Corollaries}

We get a very short corollary:
\begin{corollary}
Suppose $|A|=|B| = 2^{d} - 2^{m}$.  Then there exist node-disjoint paths
from $A$ to $B$ on a pair of butterflies.
\end{corollary}
\begin{proof}
Use Lemma~\ref{set complement} and Theorem~\ref{powers of two} on 
the complements of $A$ and $B$. 
\end{proof}

We can use Theorem~\ref{powers of two} to give us information about
a kind of rearrangeability on sufficiently small input and output sets.
\begin{corollary} \label{mini-rearrangeability}
Suppose we have a pair of $d$-dimensional butterflies.
Suppose that there is a path between each node on layer 
$\lfloor d/2 \rfloor$ (in the left butterfly) and each node
on layer $2d - \lfloor d/2 \rfloor$ (in the right butterfly).
Then if we select any input set $A$ and output set $B$ with
$A=B \leq 2^{\lfloor d/2 \rfloor}$,  and any permutation $\rho$ from
$A$ to $B$, there exists a collection of node-disjoint paths
from $A$ to $B$ such that for every $a \in A$, the path from $a$
ends at $\rho(a)$.
\end{corollary}
\begin{proof}
If the corollary holds when $A=B = 2^{\lfloor d/2 \rfloor}$, then,
by using dummy packets to make up the difference, the corollary 
holds for $A=B \leq 2^{\lfloor d/2 \rfloor}$.  So, suppose that 
$A=B = 2^{\lfloor d/2 \rfloor}$.  We can use the same argument in
Theorem~\ref{powers of two} to split the packets until there is one
packet on each $\lceil d/2 \rceil$ dimensional sub-butterfly.
By the assumption in the corollary, the resulting connectivity graph
is a complete bipartite graph on \emph{all} nodes, so we can select
node-disjoint paths between the path originating at any $a$ and
send it to the path terminating at $\rho(a)$.
\end{proof}

Note that if we have a pair of standard butterflies, the corollary holds.
Also, suppose we have a pair of layer-permuted butterflies.
Suppose further that we insist that
\begin{itemize}
\item  if $i \leq \lfloor d/2 \rfloor$, then 
$\pi(i) \leq \lfloor d/2 \rfloor$ (where $\pi$ is the left layer permutation)
 on the left butterfly, and
\item
 if $i \geq \lceil d/2 \rceil$, then 
$\sigma(i) \geq \lceil d/2 \rceil$ (where $\sigma$ is the right layer 
permutation) on the right butterfly.
\end{itemize}
(In other words, we permute the layers but don't send any layer from the
left half of the butterfly to the right half.)
Then Corollary~\ref{mini-rearrangeability} holds.

\section{The General Case}

Proving node-disjoint 
subset routing for an arbitrary input and output set (of the same 
size) is somewhat more challenging.  However, 
for pairs of layer-permuted butterflies,
the same basic approach
from Theorem~\ref{powers of two} works.  Looking at the proof, there are
two parts: first, we split the packets into a number of sub-butterflies,
until we have one packet per sub-butterfly.  Then, we view the 
problem as an
exact matching problem on a particular bipartite graph, and show that
a matching exists.  

The proof for the general case runs the same way.  In order to find 
a matching, it's clearly necessary that each connected component
of the bipartite connectivity graph has as many packets 
on the left side as on the right.  
In the next section, I'll prove that this condition (roughly speaking)
is sufficient for the existence of a matching on the 
connectivity graph.
But assuming for now that it holds, we can prove the main result:
\begin{theorem} \label{main}
For any input set $A$ and output set $B$ on a pair of $d$-dimensional 
layer-permuted butterflies,
such that $|A|=|B|$, there exist node-disjoint paths from $A$ to
$B$.
\end{theorem}

\begin{proof}
If $|A| = 2^{d}$, then we are done, by (for example) 
Theorem~\ref{powers of two}.  So throughout, we can assume that
$|A| < 2^{d}$.
Suppose that, in binary, $|A|=b_{m}b_{m-1}\cdots b_{1}$, where
$m\leq d$.  I will prove the lemma by induction on $m$.
The exact statement that I will be inducting on is: 

\emph{
Over the course of $m+1$ steps, 
we can recursively split the packets over the sub-butterflies, 
so that if sub-butterfly $x*$ has $p$ packets in it, then $x0*$ will
have $\lceil p/2 \rceil$ or $\lfloor p/2 \rfloor$ packets, and $x1*$ will
have $\lfloor p/2 \rfloor$ or $\lceil p/2 \rceil$ packets, respectively.
The same holds on the right butterfly.
(There will then be 0 or 1 packets in each sub-butterfly on level $m+1$
and level $2d-m-1$).  We can then select a matching between the 
sub-butterflies giving us node-disjoint paths from $A$ to $B$.
}

First, the base case: if $m$=0 or 1, then we are done (by 
Theorem~\ref{powers of two}).

Next, the inductive step.  Fix $m$ and assume the theorem holds
for all $m'<m$.  Let us try to reproduce the proof of 
Theorem~\ref{powers of two} with $2^{m+1}>|A| \geq 2^{m}$ 
packets to see where complications arise.  If $|A|\neq 2^{m}$, then
we will not be able to divide the packets evenly in half at every
sub-butterfly for $m$ steps.  A sub-butterfly $x*$ may have an odd number of
packets, so we must send the ``extra'' packet either to $x0*$ or $x1*$.
I will refer to this choice (the ``0'' or ``1'') as the \emph{rounding
decision}.  Note that there is no actual packet that is distinguished
as the ``extra'' one-- there's just a surplus of one more packet that
either goes to $x0*$ or $x1*$.  But it's helpful to imagine that one
of the packets is the extra one when describing the paths.

Choose $A' \subseteq A$ and $B' \subseteq B$ such that
$|A'|=|B'|=|A|-2^{m}$.   Let $m'$ be the integer such that 
$|A'|=b_{m'}b_{m'-1}\cdots b_{1}$ (so, $m'<m$).  
By induction, we can find node-disjoint paths
from $A'$ to $B'$.  The information that we keep from the
induction is not the actual paths themselves.
Instead, we keep the rounding decisions that every sub-butterfly
makes.  Note that even after
we've split $A'$ until there's only 1 packet per sub-butterfly,
we are still splitting with an extra packet; it's just that if
$p=1$, then $\lceil p/2 \rceil =1$, and $\lfloor p/2 \rfloor=0$.
Hence, the almost exact recursive splitting part of the inductive
hypothesis holds not just for the first $m'$ steps, but for the
first $m$ steps.  We need to keep this rounding
information, too.

Consider, now, the original sets $A$ and $B$.  Using 
Lemma~\ref{simple splitting} recursively for $m-1$ steps
on the right and left
butterflies, we can send the ``extra'' packet on each sub-butterfly
the same way on level $k<m$ as we did when routing $A'$ to $B'$.  
To do this, we need to know that the same sub-butterflies have an
odd number of packets in them.
Observe that if a sub-butterfly on level $k$ that has $t$ packets in it
in the $(A', B')$ case, then it  
has $t+2^{m-k}$ in the $(A,B)$ case.  
As long as $k<m$, then
$t$ and $t+2^{m-k}$ have the same parity; therefore, extra packets
exist in the same sub-butterflies.  When we reach step $m$, all
the $m$-level sub-butterflies that had one packet in them in the $(A',B')$
case now have 2 packets, and all the sub-butterflies that had no
packets now have 1.

We shift now to the matching problem on the sub-butterfly connectivity graph.
Consider one connected component of the connectivity graph.
Every node on the left hand side represents a sub-butterfly with
one or two packets on it, as does
every node on the right hand side.  By induction, the total number of
packets on each side is the same.  (If they weren't, the packets in
the $(A',B')$ case couldn't match up.)
Using Lemma~\ref{matching} of the 
next section, we can split the packets over level $m$ (and $2d-m$)
to get 0 or 1 packet per sub-butterfly, with the same number on the 
LHS and RHS of each of the connected components.  By 
Lemma~\ref{connected components} each component is completely connected,
and we're done. 
\end{proof}

Note that since we're using the direction of the ``extra'' packet, rather
than any particular path, the actual packets going into the upper or
lower sub-butterflies are not necessarily the same between the $(A',B')$
case and the $(A,B)$ case.  In particular, $A'$ will not necessarily still
be routed to $B'$.

\section{The Matching Lemma}
\begin{lemma} \label{matching}
Suppose we have a pair of $d$-dimensional layer-permuted butterflies.
Consider its $q$-dimensional
sub-butterfly connectivity graph, where the sub-butterflies
reside on layer $m=d-q$, and $2d-m$.  Suppose
that each node has 1 or 2 packets
on it.  Finally, assume that there are the same number of packets
on the LHS and the RHS of each connected component.

Then, when passing from the $q$-dimensional connectivity graph to the 
$(q-1)$-dimensional connectivity graph, we can send each packet to
a different sub-butterfly such that each connected component has
the same number of packets on the LHS and the RHS.
\end{lemma}

\begin{proof}
Since the behavior of the two-packet sub-butterflies is determined
(one packet goes to $x0*$, one to $x1*$),
this proof will eventually come down to making the correct 
rounding decision for the sub-butterflies with single packets.

There are three cases we have to consider, reflecting the three
possible behaviors of the connectivity graph as outlined on 
page~\pageref{three cases}.

\textbf{Case 1:  (No reused dimensions)}  If the connected components don't
split between the $q$ and $q-1$ dimensional sub-butterflies,
then the lemma is trivially true.

\textbf{Case 2:  (One reused dimension)}  Suppose that each connected component
splits into two connected components.  Consider one connected component
$C$ in the $q$-dimensional connectivity graph that splits into 
$C_{0}$ and $C_{1}$ in the $q-1$-dimensional connectivity graph.

Suppose that there are
$x$ nodes in $C$ with two packets on them, and $y$ nodes with one 
packet on them.  We must send $x$ packets to $C_{0}$ and $x$ to $C_{1}$
on both the left and the right sides
because the behavior of two-packet sub-butterflies is determined.
We can send the $y$ packets from one-packet nodes to either component;
we simply send
$\lfloor y/2 \rfloor$ to $C_{0}$ and $\lceil y/2 \rceil$ to $C_{1}$
on both the left and the right sides.  Then the lemma holds.

\textbf{Case 3:  (Two reused dimensions)} Suppose that each connected
component in the $q$-dimensional connectivity graph splits into four
connected components, e.g. $C$ splits into $_{0}C_{0}$, $_{1}C_{0}$, 
$_{0}C_{1}$, and $_{1}C_{1}$.  

Let us calculate how many
of the packets from the 2-packet butterflies arrive in each of these
splintered components.  

If we have a sub-butterfly $x*$ on level $m$ with 2 packets in it, then 
we must send exactly 1 packet to $x0*$ and one to 
$x1*$.  I will refer to these packets as \emph{constrained} packets.  
(By contrast, if a sub-butterfly $x*$ on level $m$ has only 1 packet in it,
we can send the packet either to $x0*$ or $x1*$; such a packet is 
a \emph{free} packet.)  We shift our view back to the corresponding
connectivity graph.
Let us label the number of constrained packets on each side of each 
$_{i}C_{j}$.
Observe first of all that because constrained packets come in pairs,
for a fixed $i=0$ or 1, there are as many constrained
packets on the LHS of $_{i}C_{0}$ as of $_{i}C_{1}$, and similarly
as many on the RHS of $_{0}C_{i}$ as of $_{1}C_{i}$.  
Let the number of packets on the LHS of
$_{0}C_{0}$ be $a_{1}$, and the number of packets on the LHS of $_{1}C_{0}$
 be $a_{2}$.
Let the number of packets on the RHS of 
$_{0}C_{0}$ be $b_{1}$, and the number of packets on the RHS of $_{0}C_{1}$
 be $b_{2}$.
See Figure~\ref{constrained}. 
\begin{figure}[ht] 
\centerline{\psfig{file=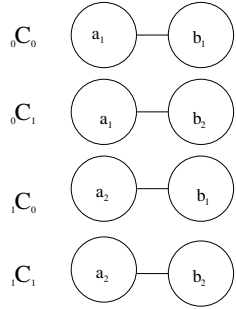,height=1.6in} \hspace{.2in}  } 
\caption{The number of constrained packets} 
\label{constrained}
\end{figure}

Observe that, since each sub-butterfly in layer $m$ has
either one free or two constrained packets, then the number of
packets on the LHS is 
\[\# \mbox{free}_{LHS} + \# \mbox{constrained}_{LHS}=
2^{m}+\frac{1}{2}(\# \mbox{constrained}_{LHS})\]
Since the analogous
equation holds on the RHS, and since the total number of packets are equal,
we get that 
\[2^{m}+\frac{1}{2}(\# \mbox{constrained}_{LHS})=
2^{m}+\frac{1}{2}(\# \mbox{constrained}_{RHS}),\]
so there's the same total number of constrained
packets on the RHS and the LHS.  Therefore, adding up the 
constrained packets in Figure~\ref{constrained} and dividing by two, we get
\[a_{1} + a_{2}=b_{1}+b_{2}\]
Also, any particular $a_{i}$ or $b_{i}$ can't be larger than 
$2^{m-1}$, so 
\[a_{i}, b_{i} \leq 2^{m-1}\]
Due to symmetry, we can assume w.l.o.g. that $a_{1}\geq a_{2}$,
$b_{1} \geq b_{2}$, and $a_{1}\geq b_{1}$.  Putting this together,
we can assume that
\[2^{m-1}\geq a_{1} \geq b_{1} \geq b_{2} \geq a_{2} \geq 0\]

Generally speaking, $a_{i}\neq b_{j}$, so there will not be the 
same number of constrained packets on the RHS and LHS of each
connected component of Figure~\ref{constrained}.  However, we
still have the free packets to allocate.  The situation is as drawn in
Figure~\ref{free 1}.
\begin{figure}[ht] 
\centerline{\psfig{file=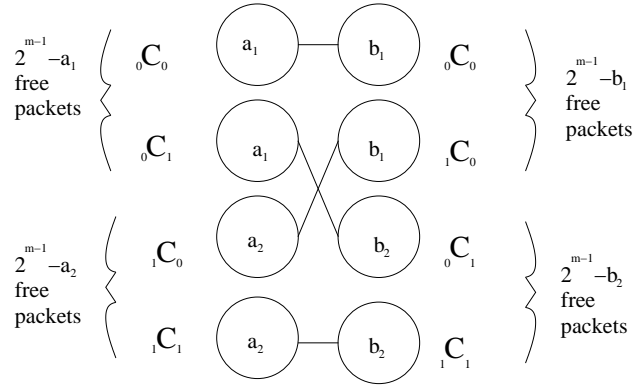,height=2in} \hspace{.4in}  }
\caption{The number of free and constrained packets} 
\label{free 1}
\end{figure}
Since we assume that $a_{1}\geq b_{1}\geq b_{2} \geq a_{2}$, then 
in order to balance the packets on the LHS and RHS, we have to
add packets as in Figure~\ref{free 2}.
\begin{figure}[ht]
\centerline{\psfig{file=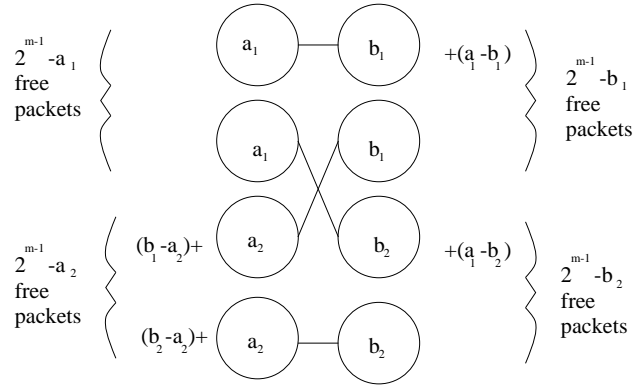,height=2in} \hspace{.4in}  }
\caption{Adding free packets to balance the bipartite graph}
\label{free 2}
\end{figure}
We have to show that there are enough free packets to add.
There are three inequalities to check.  First,
for $_{1}C_{0}$ and $_{1}C_{1}$ on the LHS,
let us calculate how many free packets are required.
\[(b_{1}-a_{2})+(b_{2}-a_{2}) = b_{1}+b_{2}-2a_{2}=a_{1}+a_{2}-2a_{2}
=a_{1}-a_{2}\]
Now, $a_{1}\leq 2^{m-1}$, so we need no more than $2^{m-1}-a_{2}$ 
free packets, which we have.
For the other two cases, (namely $i=0$ and $i=1$), note that
\[a_{1}-b_{i}\leq 2^{m-1}-b_{i}\]
and in each case, there are $2^{m-1}-b_{i}$ free packets.
So, in all cases, we can use a subset of the free packets to
make the total number of packets on the RHS and LHS equal.
Since all the remaining unmatched free packets on the left are
connected to all the unmatched free packets on the right, we can
choose an exact matching to match these packets, send them to
the appropriate connected component,
and we're done.
\end{proof}

\section{Conclusion}
Are all pairs of butterflies superconcentrators?  Or only the layer-permuted
ones?  It's certainly natural to conjecture that the stronger statement is
true.  As a piece of support, Theorem~\ref{powers of two} can be extended
to prove that any pair of butterflies is a $(2^{m}+1)$ concentrator.
Unfortunately, the pathological cases (from unusual butterfly isomorphisms)
make the general analysis more complicated than I could solve.

The concentration and superconcentration 
results in this paper all spring from a splitting
and matching approach.
This method holds out a tantalizing suggestion of a
proof of the rearrangeability of pairs of butterflies.
Theorem~\ref{main} can be viewed as follows: if we number each
input and output node 0 or 1, and have the same number of 
zeroes among the inputs and outputs, we can route a permutation
that sends $0\rightarrow 0$ and $1\rightarrow 1$.  
Suppose we labelled the input and output nodes
0,1,2, or 3, with the same size restraints.  The proofs above seem
likely to apply to this case, too.  If we could just continue
doubling the number of labels up to $N=2^{d}$, we'd have proved
rearrangeability.  Getting the proofs to work for an arbitrary
$2^{m}$ seems pretty challenging, though.

Another natural network to try these methods on is the hypercube.
Typically, rearrangeability on the hypercube requires that each
edge is used at most once, ever, and concerns edge-disjointness, rather
than node-disjointness.  A result analogous to 
Theorem~\ref{main} would be more likely to apply to a hypercube
that uses each edge at most once per time step, but possibly multiple
times over several time steps.  However, edge-disjointness might
be strengthened to node-disjointness.  Unfortunately, the translation
to a hypercube is not trivial.

Proving that a graph is a superconcentrator can also be viewed as a
max flow/ min cut problem; thus, Theorem~\ref{main} can be viewed as 
saying that
for any collection of $k$ input and $k$ output nodes, it is 
necessary to delete at least $k$ edges to prevent any (single-pass)
paths from the input to the output sets.  One might optimistically hope
that these results might translate to other max flow problems, at least
on switching networks.

On a possibly more practical note, it's interesting to observe that the
Theorem~\ref{main} makes use of the size of the 
input set, rather than the set itself (i.e.\ $|A|$, not $A$).  It follows that
once you calculate the rounding decisions for a particular sized
input set, the same rounding decisions solve the problem for all
input sets of the same size.  This also suggests another method for
proving concentration results.

A version of this problem was originally suggested to me by Bruce Maggs, 
in the 
context of routing on faulty butterflies.  I would like to thank him for
many helpful discussion on it.

\bibliography{butterfly}

\begin{thebibliography}{1}

\bibitem{Cam}
Hasan \c{C}am.
\newblock Rearrangeability of (2n-1)-stage shuffle-exchange networks.
\newblock Forthcoming.

\bibitem{Hwang}
Frank~K. Hwang.
\newblock {\em The mathematical theory of nonblocking switching networks}.
\newblock World Scientific, Singapore, 1998.

\bibitem{Leighton}
F.~Thomson Leighton.
\newblock {\em Introduction to Parallel Algorithms and Architectures: Arrays
  $\bullet$ Trees $\bullet$ Hypercubes}.
\newblock Morgan Kaufmann Publishers, Inc., San Mateo, CA, 1992.

\bibitem{Pippenger}
Nicholas Pippenger.
\newblock Superconcentrators.
\newblock {\em SIAM J. Comput.}, 6(2):298--304, June 1977.

\bibitem{optical_networks}
Yuanyuan Yang, Jianchao Wang, and Yi~Pan.
\newblock Permutation capability of optical multistage interconnection
  networks.
\newblock {\em Journal of Parallel and Distributed Computing}, 60:72--90, 2000.

\end{thebibliography}
\bibliographystyle{plain}

\end{document}